\documentclass[12pt, reqno]{amsart}
\usepackage{cite}
\usepackage{amsmath,amsfonts,amsthm,amssymb,color,hyperref}
\usepackage{lscape}
\usepackage{amsmath,bm}
\usepackage{setspace}
\usepackage{adjustbox}
\usepackage{booktabs}
\usepackage{caption}
\usepackage{subcaption}
\usepackage{algorithm}
\usepackage{algpseudocode}

\usepackage{mathrsfs}
\usepackage[T1]{fontenc}
\usepackage{graphicx, float}
\usepackage[dvipsnames]{xcolor}
\usepackage{cleveref}
\usepackage{multirow}
\usepackage{siunitx} 
\usepackage{placeins}
\usepackage{easyReview}
\usepackage{typed-checklist}
\makeatletter
     \def\section{\@startsection{section}{1}%
     \z@{.7\linespacing\@plus\linespacing}{.5\linespacing}%
     {\bfseries
     \centering
     }}
     \def\@secnumfont{\bfseries}
     \makeatother
\usepackage{graphicx}
\usepackage{xcolor}
\makeatletter
    \def
    \paragraph{
        \@startsection{paragraph}{4}%
        \z@\z@{-\fontdimen2\font}%
        {\normalfont\bfseries}
    }
\makeatother

\newcommand{\E}{\mathbb E}

\newcommand{\Pb}{\mathbb P}
\newcommand{\cf}{\mathcal F}

\newtheorem{theorem}{Theorem}[section]

\theoremstyle{definition}

\theoremstyle{remark}
\newtheorem{remark}{Remark}
\numberwithin{equation}{section}
 \DeclareGraphicsExtensions{.pdf,.png,.jpg,.tiff}

\restylefloat{figure}

\begin{document}
\title[Inference for a stochastic generalized logistic differential equation]{Statistical inference for a stochastic generalized logistic differential equation}

\author[F. Baltazar-Larios]{Fernando Baltazar-Larios} 
\address{  {\it F. Baltazar-Larios}.
		Facultad de ciencias - Universidad Nacional Aut\'onoma de M\'exico.  
		Ciudad Universitaria, Ciudad de M\'exico}

\author[F. Delgado-Vences]{Francisco Delgado-Vences}
\address{{\it F. Delgado-Vences\footnote{Corresponding author} }. Conacyt  Research Fellow - Universidad Nacional Aut\'onoma de M\'exico. Instituto de Matem\'aticas, Oaxaca, M\'exico}
\email{delgado@im.unam.mx}

\author[S. Diaz-Infante]{Saul Diaz-Infante} 
\address{{\it S. Diaz-Infante}. Conacyt Research Fellow - Universidad de Sonora - M\'exico}

\textit{}\author[E. Lince]{Eduardo Lince Gomez} 
\address{  {\it E. Lince G\'omez}.
		Facultad de ciencias - Universidad Nacional Aut\'onoma de M\'exico.  
		Ciudad Universitaria, Ciudad de M\'exico}

\begin{abstract}

The aim of this paper is to study a method to estimate three parameters in a stochastic generalized logistic differential equation (SGLDE) with a random initial value. We assume that the intrinsic growth rate and a shape parameter are constants but unknown, then we use maximum likelihood estimation to estimate these two parameters.  In addition, we prove that these two estimators are strongly consistent as the observation time window increases, i.e., as $T$ grows. The diffusion parameter is estimated using the quadratic variation,  which we know is an unbiased estimator. The method was tested by assigning fixed values to the three parameters under two different scenarios of available data (complete and incomplete data). In the incomplete case, we use an  Expectation Maximization (EM) algorithm. The solutions were simulated using the closed solution of the SLDE. Afterward, our method was employed to estimate the parameters and to investigate a form of empirical convergence toward the true value. 


\end{abstract}

\maketitle

\medskip\noindent


\allowdisplaybreaks

\section{Introduction and contribution}

The deterministic generalized logistic differential equations (GLDE), also known as the Richards equation, is an extension of the logistic  model, having originally been employed to model biological growth. A deterministic GLDE version was developed by \cite{Rich1959} where the authors presented an extended von Bertalanffy  growth function to model plant data. The model incorporates an additional parameter, known as the shape parameter. This parameter enables the Richards equation to be rendered equivalent to the logistic, Gompertz, or monomolecular equations as outlined in \cite{fr-th-84}. The variation of the shape parameter allows the point of inflection of the curve to be positioned at any value between the minimum and the upper asymptote.

In this paper, we study the problem of estimating the three parameters for a stochastic generalized logistic differential equation (SGLDE). This problem is quite challenging, as we will show, due to the interdependence of the two drift parameters. Therefore, an efficient method is needed to estimate both of them simultaneously.

The SGLDE has been extensively researched by various authors (e.g., Panik  \cite{panik}). Apart from capturing the average characteristics of the data, this model also addresses the individual variability of the population.

\paragraph{Our contribution} 
The deterministic Richards equation is highly relevant for fitting growth data, yet there remains a scarcity of information for calibrating stochastic extensions. To address this gap, we propose a method for adjusting a stochastic generalized logistic model using consistent estimators and the EM algorithm.

The method of estimation we use is as follows. Initially, we estimate the diffusion parameter using quadratic variation. Afterward, we  use the Girsanov theorem to obtain the Radon-Nykodim derivative of the measure generated by the solution. This enables us to derive the Maximum Likelihood estimator for the parameters in the drift component. Furthermore, we establish the strong consistency of these estimators. 

We validate the proposed method by conducting simulations with two scenarios of available data. Specifically, we will consider two observation cases. The first scenario entails the completed data wherein continuous observations of the solution's paths in SGLDE are available, which can be regarded as an unrealistic case. The second scenario represents a more realistic situation where we only have observations of the path at discrete points in time. Certainly, in the second case, we can conceive this scenario as a derivative of incomplete observation in the continuous case. We suggest using a stochastic version of the EM algorithm to fill the gap between the observations and combine it with the estimators to complete the data, in some sense.

We hereby present several simulation examples aimed at calibrating our methodology. These examples incorporate various values for all parameters under consideration.\\

\paragraph{Logistic deterministic equation} 
This subsection is devoted to a brief review of the logistic equation. 

The logistic growth model, an ordinary differential equation defined as 
\begin{equation}
    \label{eqn:logistic_verhulst}
    X^{\prime}(t) = 
    \alpha X(t)
    \left[
    1-\frac{X(t)}{K}
    \right],
\end{equation}
is one of the most important models in population dynamics.
Proposed by Verhulst in 1838 (see e.g.,  \cite{Vogels1975})
it remains a classic model for studying
population growth and as a cornerstone in more complex structures.
The scope of this application encompasses various fields, such as 
economics, finance, chemistry, and biology, among others.

 After Verhulst first presented it as a contribution to the law of
    population growth, the well-known  
exponential growth model
$
    X'(t) = \alpha X(t)
$; the model has been rediscovered and
    generalized in several fields. Richards conducted in 1959 a 
    well-known generalization \cite{Rich1959}.
    In this piece of research, Richards adds the shape parameter $m$ to
    extend the von Bertalanffy  growth function to fit plant data. Richards
    reformulates the  model
    \cref{eqn:logistic_verhulst} as 
    \begin{equation}\label{det-GLDE}
        \left.
            \begin{aligned} 
                X'(t)&= 
                    \alpha X(t)
                        \Bigg[
                            1 - 
                            \Bigg(
                                \frac{X(t)}{K}
                            \Bigg) ^ m
                        \Bigg],
                    \quad t>t_0
                \\
                    &X(t_0)= x_0,
                    \quad
                  x_0\in(0,1).
            \end{aligned}
        \right\}
    \end{equation}

According to \cref{det-GLDE}, if the initial population size $x_0$ 
is positive and less than the carrying capacity $K$, then  $X(t)$ could be interpreted as 
the population size at time $t$, that grows at rate  $\alpha$ and is limited by 
a finite number $K$ of natural resources.
By introducing the new parameter $m$, the logistic sigmoid shape of the 
previous model \eqref{eqn:logistic_verhulst} achieves more plasticity and 
expands its fitting range. Indeed, for values of $m$ in $(0,1)$, the sigmoid 
function tends towards a linear activation function; while for $m>1$, the 
sigmoid shape becomes similar to the Heaviside step function. Moreover, 
according to \cite{Rich1959} (see also  \cite{fr-th-84}), this generalization encompasses different 
values of $m$, which accounts for well-known growth modes such as the 
monomolecular, auto-catalytic, and Gompertz. 

Certain studies suggest that the logistic Richards' model yields more 
precise results when fitting real data. For instance, Qin et al. \cite{Qin2018} observed 
better performance with the Richards' model \eqref{det-GLDE} than with the 
Verhulst model \cref{eqn:logistic_verhulst} when calibrating a database of 
monthly smoothed international sunspot numbers.

\Cref{det-GLDE} is a particular case of the Bernoulli differential 
equations. When we substitute $u$ by $X^{1-m}$ in \cref{det-GLDE}, 
it becomes a linear equation that we can solve using an integral factor (see cf. \cite{br-go}). This argument leads us to the following solution
\begin{equation}
    X(t)
        =\frac{K}{(1+Q e^{-\alpha m (t-t_0)})^{1/m}},
    \qquad  
    Q:=
    \left(
    \frac{K}{x_0}
    \right)^m - 1.
\end{equation}

Assuming this model represents biological growth, then we can fix parameter $K$ using 
prior biological experience or just by considering the maximum value of actual 
observations. Thus, without loss of generality, we can rescale the dynamics of 
the equation in terms of this 
parameter. When we rescale equation \cref{det-GLDE} in terms of $K$, we 
get dynamics that fall within the interval $(0,1)$, with two fixed points at 
$X^{\star}=0$ and $X^{\star \star}=1$. The point $X^{\star}$ is unstable, while 
the point $X^{\star \star}$ is asymptotically stable for any initial condition 
$x_0 \in (0,1)$.\\

\paragraph{Literature Review of non-deterministic GLDEs}
Several contributions have reported formulations that are 
vaguely related to the model presented in equation \eqref{SGLDE}. 
From a theoretical standpoint, Suryawan \cite{su} assesses various qualitative aspects of the solution, such as long-time behavior and noise-induced transition, 
using the framework of Ito's stochastic calculus. Cortes et al. 
\cite{Cortes2019} study a random logistic type differential equation and they obtain the probability density of the underlying solution 
processes using the Karhunen–Lo\`{e}ve expansion and other transformations. 
Most of the literature focuses on the  model \eqref{SGLDE} with $m=1$. Other 
authors such as \cite{Liu2013, Braumann2008} document the qualitative 
behavior of the solution process using a different form of diffusion term. 
Nevertheless, \cite{Schurz2007} 
presents a theory and exhaustive numerical analysis for other linear 
generalizations in the drift and nonlinear terms in the diffusion term.\\

Inference for SDEs has become a well-established research area in modern times. For comprehensive coverage, one can consult, for instance, the works of \cite{iacus}, and \cite{panik}, as well as the review article by \cite{cr-he-li-sc}, which provides valuable insights into this field, along with references for further exploration.

To the best of our knowledge, only \cite{Bevia2023} provides details on 
parameter calibration using real data for a non-deterministic GLDE;
however, the authors use a model based on a random differential 
equation (RDE) instead of a stochastic differential equation;
which are different. Indeed, in the RDE they studied, the source 
of randomness comes from the initial condition and the growth rate,
which is assumed to be a bounded random variable rather than a stochastic process.

In \cite{del-bal}, two of the authors of this manuscript have studied statistical inference for a stochastic logistic differential equation driven by an affine noise; that means they have considered the model \cref{SGLDE} with $m=1$ (see below). Indeed, they have used the same type of ideas we use in our work to obtain the maximum likelihood estimator (MLE) of the parameter $\alpha$. They have also shown the consistency and asymptotically normality of the MLE. In addition, they have applied the model to actual data. That work is the inspiration to the model extension considered in this paper, that is the case $m$ not necessarily equal to $1$. Thus, our contribution seems to be highly pertinent and relevant.\\


This paper is organized as follows. In section \ref{sec-SGLDE} we introduce the model we consider in this work, which is given by an SDE driven by a particular multiplicative noise, which is usually called affine noise . In Section \ref{sec-EstPAR} we study the MLEs for the two parameters in the drift term of the SGLDE, we prove that these are strongly consistent estimators and the estimator for the diffusion parameter by using the quadratic variation process. Afterward, we study the MLE in the two scenarios described above. Section \ref{simul_study} is devoted to showing numerical experiments using the results and methods from Section \ref{sec-EstPAR}.  Section \ref{sec-Conclusions}  has the concluding remarks on the work and future related work, while in the Appendix we show the consistency of the MLEs.

\section{A stochastic generalized logistic differential equation}
\label{sec-SGLDE}
In this section, we introduce the SDE that serves us as a stochastic model.
  
\paragraph{Stochastic extension}
The incorporation of uncertainty into mathematical models yields enhanced performance when applied to real-world data. Then, in the case of deterministic GLDE recently appeared works that extend this model to
random differential equations and stochastic differential equations. For further reading, we refer the reader to
\cite{Cortes2019}, \cite{Bevia2023}, \cite{Schurz2007}, \cite{su} among others.

In this paper, we focus on the stochastic logistic model driven by an affine noise: 
\begin{equation}\label{SGLDE}
  \left.
        \begin{aligned}
            dX(t) &= 
                \alpha X(t)\Bigg[1-\Bigg(\frac{X(t)}{K}\Bigg)^m\Bigg]dt 
                + 
                \sigma X(t) dB(t),
                \qquad t>t_0,\\
                X(t_0)&= x_0.
        \end{aligned}
        \right\}
\end{equation}
Here, $x_0 \in (0, 1)$ is the initial condition that we could assume that is a bounded absolutely continuous random variable, and $\{B(t)\}_{t\geq t_0}$ 
is a standard Brownian motion and the solution process $X(t)$ 
is understood under the usual conditions. That is, we assume that
solution $X$ and $x_0$ are defined on a common filtrated 
probability space 
$
    (
        \Pb, 
        \Omega,
        \cf
    )
$, with filtration $\cf =\{ \cf \}_{t\geq 0}$ generated by the Brownian process
$B(t)$. \Cref{SGLDE} is a Itô-stochastic differential equation (SDE) driven by a
multiplicative noise.



By using It\^o's formula we obtain the closed solution to \eqref{SGLDE}
(see \cite{su} for further reading) 

\begin{equation}
    \label{sol-SGLDE}
    \begin{aligned}
         X(t)=x_0 \,& 
            \exp\Big[  
            \Big(
                \alpha - \frac{1}{2}\sigma^2
            \Big) t 
            + \sigma B(t)\Big]
        \\
            &\, 
            \times 
            \Bigg[
                1 +  
                \Big(
                    \frac{ x_0}{K}
                \Big) ^ m
                \alpha m 
                \int_0^t 
                \exp
                \Big[
                    \Big(
                        \alpha - \frac{1}{2} \sigma^2
                    \Big)s 
                + \sigma 
                    B(s)
                \Big]  ds 
            \Bigg] ^ {-1/m}.
    \end{aligned}    
\end{equation}
We observe that the solution is always positive for all $t\ge 0$.

It is not difficult to see that 
$$
\lim_{t\rightarrow\infty}  \E [X(t)]= K.
$$

Since we are interested in the model for biological growth, 
then we can restrict ourselves to the case $K=1$.

\paragraph{Aim}

We assume that the parameters $\alpha,m$, and $\sigma$ are constants but unknown. 
 In this manner, our first goal is to fit the model to actual data; thus we need to estimate three parameters: 
$\alpha,m$, and $\sigma$. We will fit the model to data with two scenarios of available information. The first scenario will involve complete information, where the path along time is provided for the model fitting process. In the second scenario, we will consider the case of incomplete information, where only a subset of the trajectory is available.

\section{Model and data}\label{sec-EstPAR}
    For simplicity in the notation, we will denote the parameter vector
    $\theta=\{\alpha,m\}$, which 
    is unknown and we denote by $\theta_0$ the true value of the parameter vector. We consider the estimation of $\theta$ based on observations of the process solution of \eqref{SGLDE} sampled in the time interval $[0,T]$ for continuous and discrete cases.
\subsection{Continuous (complete) observation}\label{sec:con_cas}
For $T>0$, we will say that  {\it continuous observation} within the interval  $[0, T]$ is present if we possess a dataset 
$$
\big\{X(t_0)=x_{t_0},X(t_1)=x_{t_1},\ldots,X(t_n)=x_{t_n}\big\}
$$ 
for $0=t_0<t_1\ldots,t_n=T$ with $n$ large enough such that, for $i=1,\ldots,n$, $\Delta_n:=t_i-t_{i-1}=\tfrac{T}{n}$  is close to zero.

Thus, as a first scenario, we assume continuous observation of paths of the diffusion given by \eqref{SGLDE} in the time interval $[0,T]$.  Based on these observations, we provide estimators for $\theta$.

In the first step, we estimate $\sigma$ via quadratic variation and, given the estimator of $\sigma$, we find the MLEs for the other two parameters. We also prove the asymptotic consistency of the MLEs. We could assume that the initial condition $x_0$ is a random variable with some density and that this density is the same for all possible values of $\theta$. However, For illustrative purposes, we will use a fixed value for each simulation in the numerical experiments presented here.


\subsection*{Estimator for $\sigma$} 

Since, for all $T>0$ the quadratic variation of the diffusion process
$\bm{X}=\{X_t\}_{t\in[0,T]}$, solution of equation \eqref{SGLDE}, is given by
$$<\bm X,\bm X>_T\,:=\mbox{lim}_{n\rightarrow\infty}\sum_{i=1}^n\, [X(t_i)-X(t_{i-1})]^2=\int_0^T\sigma^2X^2(t)dt=\sigma^2\int_0^TX^2(t)dt.
$$
Then, we can estimate $\sigma$ by using the quadratic variation. Indeed, the estimator is given by  
\begin{equation}\label{QV-Sigma}
\hat{\sigma}_{QV}=\sqrt{\frac{<\bm X,\bm X>_T}{\int_0^TX^2(t)dt}}\,\approx \sqrt{\frac{2\sum_{i=1}^n(X(t_i)-X(t_{i-1}))^2}{\Delta_n\sum_{i=1}^n (X(t_i)^2+X(t_{i-1})^2)}}.
\end{equation}
\begin{remark}$\hat{\sigma}_{QV}$ from the expression \eqref{QV-Sigma} is an unbiased estimator (see \cite{wei}).\\

\end{remark}

\subsection*{Maximum likelihood estimators for $\alpha$ and $m$}

Denote by $\Pb_{\theta}$ the probability measure on the space of continuous functions $C(0,T)$ generated by $\bm X$. It is known that $\Pb_{\theta}$, and $\Pb_{\theta_0}$ are equivalent for different values of $\theta_0$, and $\theta$ \footnote{Recall that at this point $\sigma$ has already been fixed.} (see \cite{iacus} or \cite{li-sh}). Then, by the Girsanov theorem, the likelihood ratio is 
\begin{align}\label{likelihood}
L_T(\theta|\bm X)=\frac{d\Pb_{\theta}}{ d\Pb_{\theta_0}}(X) & =\exp\Bigg\{ \int_0^T \left(\frac{\alpha(1-X^m(t))}{\sigma^2 X(t)}-\frac{\alpha_0(1-X^{m_0}(t))}{\sigma^2 X(t)}\right) \, dX(t)\nonumber  \\ 
& \qquad\quad -\frac{1}{2} \int_0^T \Bigg(\left(\frac{\alpha^2(1-X^m(t))^2}{\sigma^2}-\frac{\alpha^2_0(1-X^{m_0}(t))^2}{\sigma^2}\right)\Bigg) dt  \Bigg\}.
\end{align} 

The MLEs for the parameters are commonly attained by differentiating the log-likelihood $l_T(\theta|\bm X):=\log(L_T(\theta|\bm X)$ with respect to the parameters. Thus,  the MLE for $\alpha$ is

\begin{equation}\label{MLE-alpha}
\hat{\alpha}_{ML}= \frac{1}{\int_0^T  \big(1-\big(X(t)\big)^m\big)^2   dt  }  \int_0^T \frac{ \big(1-\big(X(t)\big)^m\big)  }{X(t)}  dX(t).  
\end{equation}

Now, we calculate the derivative of $l_T(\theta|\bm X)$ with respect to $m$, and by using 
the MLE $\hat{\alpha}_{ML}$ from \eqref{MLE-alpha}, we obtain

 

 
%
\begin{equation}
    \label{MLE-m}
    \begin{aligned} 
        \frac{\partial l_T(\theta|\bm X)}{\partial m}
            &=
            \left(
                \int_0^T
                    \big(1-\big(X(t)\big)^m\big)^2   
                dt
            \right) 
            \left(
                \int_0^T
                    X(t)^{m-1} (-\log(X))
                    dX(t)
            \right)
            \\
            - &
            \left(
                \int_0^T 
                \frac{
                    \big(1-\big(X(t)\big)^m\big)
                }{X(t)}
                dX(t)
            \right)
            \left(
                \int_0^T
                    X(t)^{m}
                    \big(1-X(t)^m\big)
                    \big(-\log(X) \big)
                dt 
            \right) .
    \end{aligned}
\end{equation}

Define the function 
$g(m) := \frac{\partial l_T(\theta|\bm X)}{\partial m}$. Thus, we find the MLE $\hat{m}_{MLE}$ for $m$ by
computing a positive root of the equation   
\begin{equation}\label{eq:m}
    g(m)=0.   
\end{equation}
%

Due to the non-linear nature of the right-hand side of equation \eqref{MLE-m}, it is impractical to find an explicit form to solve equation \eqref{eq:m}. Henceforth, it is imperative to employ numerical approximations to acquire a solution. Newton's method is widely regarded as one of the most effective approaches for solving equations such as equation \eqref{eq:m}. Given sufficient regularity and a strong initial approximation, or by implementing preconditioning techniques, Newton's method can achieve quadratic order convergence. Thus, we will utilize this method to seek a positive root of equation \eqref{eq:m}.

\begin{remark}
    Observe that in Eq. \eqref{MLE-m} the MLE for the parameter $\alpha$ has been used to obtain that expression. This implies that the two MLE expressions exhibit a mutual dependence. 
\end{remark}



\subsubsection{Consistency of MLEs}
An important contribution of this work is that the estimators for parameters $\alpha$ and $m$ are strongly consistent which is proven in Theorem \ref{theo:cons}.
\begin{theorem}\label{theo:cons}
 The estimators $\hat{\alpha}_{ML}=\widehat{\alpha}_{T}$, $\hat{m}_{ML}= \widehat{m}_{T} $, given by  \eqref{MLE-alpha} and \eqref{eq:m}, are strongly consistent, that is 
  \begin{equation*}
   \lim_{T\rightarrow \infty} \begin{pmatrix}
       \hat{\alpha}_{ML}    
        \\
       \hat{m}_{ML}
    \end{pmatrix} = \begin{pmatrix}
       \alpha_0
        \\
       m_0
    \end{pmatrix}, \qquad \mbox{in probability}.
  \end{equation*}
  
\end{theorem}

The proof is deferred to Appendix \ref{Appe_1}.

\begin{remark}
    From the proof of the last theorem we also deduce  that both MLEs are asymptotically unbiased estimators.
\end{remark}

\subsection{Discrete (incomplete) case}\label{sec:dis_case}
 In this section, we consider the scenario when we have a {\it discrete observation} of the continuous process $\bm X$, i.e., we have only records at times $0=t_0<\ldots<t_k$ when $\Delta_i=t_i-t_{i-1}$, for $i=1,\ldots,k$ and $k$ is small, we denote the data by $\bm{X}^{obs}=\{X(t_0),\ldots, X(t_k)\}$. 

It is well known that when we do not have enough data, we cannot achieve accurate estimations by directly using the expressions presented in Section \ref{sec:con_cas} (cf. \cite{cr-he-li-sc} and the references therein). Our strategy to tackle this problem, and to obtain good estimators, is to see the gap between observations as incomplete information in such a manner that we complete the information and we could be in a position to use the results from Section \ref{sec:con_cas}. Thus, to achieve that we simulate paths between two observations (diffusion bridges) to consider that we have continuous observation.  To this end, we use the method that applies to ergodic diffusion processes in \cite{bl-so-14}.

 The data set $\bm{X}^{obs}$ can be conceptualized as an incomplete observation of the complete data set $\bm X$. Consequently, the EM algorithm can be utilized for likelihood-based estimation (see \cite{mc-kr-97}).
 
\subsubsection{EM algorithm}\label{sec:EM-meth}
Assume that $\bm{X}^{obs}$ is contained in the continuous observation set $\bm X$. 
Moreover, the first and last elements of $\bm X$, namely, $X(t_0)$ and $X(t_k)$, are assumed to be part of $\bm{X}^{obs}$.

The EM algorithm commences by utilizing proposed values for the initial parameters, denoted as $\theta^0$. The selection of this particular set may be considered arbitrary; however, it bears little significance as the EM algorithm ultimately achieves convergence. Then, the first proposal is computed using the expressions \eqref{QV-Sigma}, \eqref{MLE-alpha}, and \eqref{eq:m} with the incomplete observation denoted as $\bm{X}^{obs}$.

The likelihood ratio \eqref{likelihood} is a conditional exponential family (see \cite{kuch2006}). As a result, the E step in the EM algorithm could be simplified by calculating the conditional expectation of the sufficient statistics $\bm X_i=\{X(t)\}_{t=t_{i-1}}^{t_i}$, for $i=1,\ldots,k$. Then, our version of the EM algorithm is as follows.


\begin{algorithm}[H]
\caption{EM Algorithm}
\begin{algorithmic}
  \State (1) Set $n = 0$ 
  \State (2) Initialize $\theta^0$
    \State 
    (3) \textbf{E-step:}
    \State \hspace{\algorithmicindent} For $i \in\{ 1,\dots,k\}$ compute \[Y_i^n:=E_{\theta^n}[\bm{X}_i|X(t_{i-1}),X(t_{i})]\]

    \State 
    (4) \textbf{M-step:}
    \State \hspace{\algorithmicindent} 
    Using \eqref{MLE-alpha}  and \eqref{eq:m} compute the MLE\, $\widehat{\alpha}^{n+1}_{ML}$ and $\widehat{m}^{n+1}_{ML}$, with the n-th proposed trajectory \[\bm{X}_n^{prop}:=\{Y_1^n,\dots,Y_k^n\}\]
    \State 
    (5) \textbf{Update $\sigma$:}
    \State \hspace{\algorithmicindent} Compute $\widehat{\sigma}^{n+1}_{QV}$ with \eqref{QV-Sigma}
    \State (6) Set $n = n + 1$ and go to 3.

\end{algorithmic}
\label{al:em}
\end{algorithm}

The steps 3-6 of Algorithm \ref{al:em} are iterated until convergence.
In step 3, we compute $Y_i^n$ by applying the Monte Carlo method to the diffusion bridges generated between the elements of $\bm{X}^{obs}$ using the proposed parameters. Specifically, we simulate $N$ bridges given the first element $X(t_i)$ and the last $X(t_{i-1})$ for $i = 1,\dots,k$, and then we calculate the mean of each point of the $N$ bridges.

In each iteration, the parameters of the proposal will become progressively closer to the actual parameters, resulting in $\bm X^{prop}$ becoming more accurate to $\bm X$ and subsequently enhancing the proposal for the next set of parameters.

\section{Simulation study}\label{simul_study}
The numerical results of a simulation study are presented in this section. We evaluate both scenarios discussed in Section \ref{sec-EstPAR}, namely the complete and the incomplete data, and conduct a numerical experiment to test the methods.

\subsection{Experiment 1: Continuous (complete) data}
\label{sub_sec:experimtn01}
Here, we illustrate numerically the estimation and the consistency properties of 
the estimators under the unrealistic assumption of complete data within an interval
time $[0, T]$.

First, we consider a path of the solution process of \eqref{SGLDE} with parameters 
${\alpha = 1.0}$, ${m=2.0}$, and ${\sigma = \num{0.05}}$ in 
the interval time $[0,10]$ with $n=5,000$. \Cref{fig:consistency} 
illustrates the numerical convergence of the estimators to the parameters as the value of $T$ increases within the observed interval. 
This exhibits the consistency of the estimators proved in Theorem \ref{theo:cons} 
for $\hat{\alpha}_{MLE} $, and $\hat{m}_{MLE}$. It should be noted that the estimators for the parameters $m$ and $\sigma$ exhibit 
initial poor performance at the start of the observation, with improvement occurring as the observation period progresses.

\begin{figure}[H]
  \centering
    \includegraphics[width=1\textwidth]{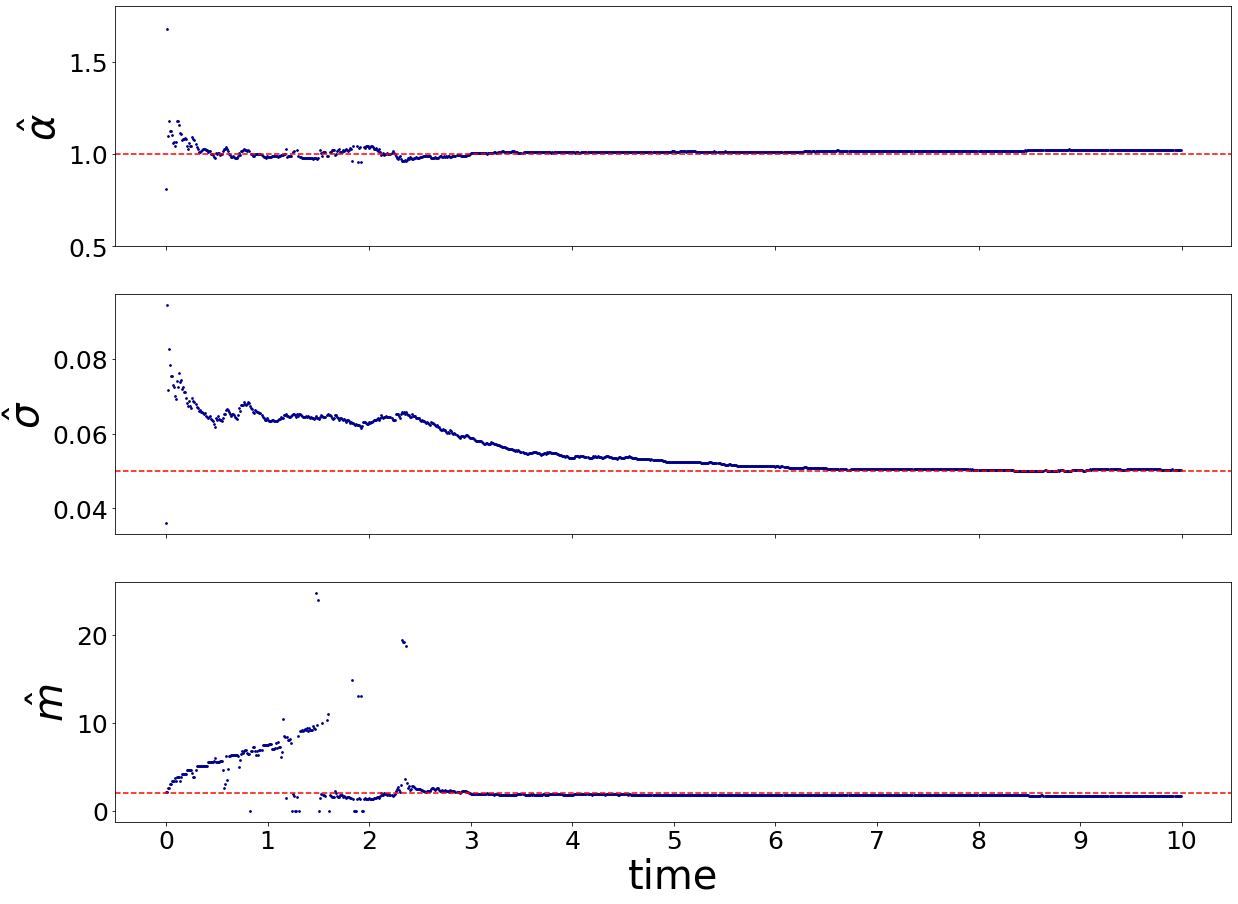}
    \caption{%
        Illustration of the consistency of MLEs $(\hat{\alpha},\hat{m})$ and the QV estimator $\hat{\sigma}$.%
    }
    \label{fig:consistency}
\end{figure}
\Cref{fig:log-error}  displays  the $\log$-error  of our estimation with complete data. That is, we 
present the natural logarithm of the absolute value of the difference between
the real value of the parameter and the estimator in the interval time $[0,10]$
for the estimators obtained from \eqref{MLE-alpha}, \eqref{MLE-m}, and \eqref{QV-Sigma}.
After conducting the experiment, we observed that the errors, when determining the values of parameters $\alpha$ and $m$, are directly proportional to the step size of numerical integration and the non-linearity of these parameters. The analysis has revealed that even with damped amplitude, the values of $\alpha$ and $m$ remain consistent for $t>\num{7.0}$. However, it is important to note that these errors still exceed the $\sigma$ log error due to the higher-order nature of the quadratic variation method used for estimating this parameter.

\begin{figure}[H]
  \centering
 \includegraphics[width=0.5
 \textwidth]{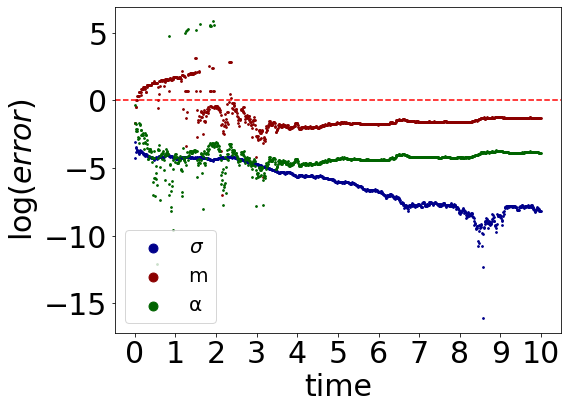}
  \caption{Log error between estimators and parameters}
  \label{fig:log-error}
\end{figure}

In Table \ref{tab:cont}, we report the average, the 95$\%$ 
quantiles, and the mean squared errors  of the estimators obtained from a series of data obtained after simulating 1000 paths in the time interval $[0,10]$ with ten thousand points each and the length of each step size $\Delta_i$ ($i=1 ,\ldots,10 000$) as $0.001$. We calculate the estimators for each of them using 
\eqref{MLE-alpha}, \eqref{eq:m}, and \eqref{QV-Sigma} for four different values of all parameters. We can observe that we obtained good estimators for all cases. 

\begin{table}[H]
    \flushleft
    \textbf{Notation:}
    \\
        \hspace{1cm} \textbf{PE} puntual estimator
    \\
        \hspace{1cm} \textbf{95Q} 95\% quantiles
    \\
        \hspace{1cm} \textbf{MSE} mean squarred error
    
    \centering
    \begin{tabular}{
        @{}%
            w{l}{0.05\textwidth}%
            w{r}{0.15\textwidth}%
            w{c}{0.15\textwidth}%
            w{c}{0.30\textwidth}%
            w{c}{0.15\textwidth}%
        @{}%
    }
        \toprule
        \\
            & \multicolumn{4}{l}{\textbf{Case 1:} $\alpha=0.7$, $\sigma=0.01$, $m=0.6$}
        \\ 
            \cmidrule(lr){2-5} 
            &
            & 
                \multicolumn{1}{c}{
                    \textbf{PE}
                }
            & 
                \multicolumn{1}{c}{
                    \textbf{95Q}
                }
            & 
                \multicolumn{1}{c}{
                    \textbf{MSE} 
                }
        \\
         \cmidrule(lr){3-3} 
         \cmidrule(lr){4-4}
         \cmidrule(lr){5-5} 
            & $\widehat{\alpha}$       
            & 
                \num{0.700 888}
            &
                (\num{0.657 515}, \num{0.752 729})
            &
                \num{6.1e-4}
        \\
            & 
                $\widehat{\sigma}$
            & 
                \num{0.011 282}
            &
                (\num{0.011 099}, \num{0.011 500})
            &
                \num{8.5e-9}    
        \\
            & 
                $\widehat{m}$
            & 
                \num{0.600 990}
            &
                (\num{0.531 013}, \num{0.671 848})
            &
                \num{1.3e-3}
        \\
        \\
        & \multicolumn{4}{l}{\textbf{Case 2:} $\alpha=0.9$, $\sigma=0.01$, $m=1$}
        \\ 
            \cmidrule(lr){2-5} 
            &
            & 
                \multicolumn{1}{c}{
                    \textbf{PE}
                }
            & 
                \multicolumn{1}{c}{
                    \textbf{95Q}
                }
            & 
                \multicolumn{1}{c}{
                    \textbf{MSE} 
                }
        \\
         \cmidrule(lr){3-3} 
         \cmidrule(lr){4-4}
         \cmidrule(lr){5-5} 
            & $\widehat{\alpha}$       
            & 
                \num{0.901 989}
            &
                (\num{0.874 677}, \num{0.933 052})
            &
                \num{2.1e-4}
        \\
            & 
                $\widehat{\sigma}$
            & 
                \num{0.011 247}
            &
                (\num{0.011 100}, \num{0.011 400})
            &
                \num{5.8e-9}  
        \\
            & 
                $\widehat{m}$
            & 
                \num{0.997 121}
            &
                (\num{0.925 817}, \num{1.070 707})
            &
                \num{1.3e-3}
        \\
        \\
        & \multicolumn{4}{l}{\textbf{Case 3:} $\alpha=1$, $\sigma=0.05$, $m=2$}
        \\ 
            \cmidrule(lr){2-5} 
            &
            & 
                \multicolumn{1}{c}{
                    \textbf{PE}
                }
            & 
                \multicolumn{1}{c}{
                    \textbf{95Q}
                }
            & 
                \multicolumn{1}{c}{
                    \textbf{MSE} 
                }
        \\
         \cmidrule(lr){3-3} 
         \cmidrule(lr){4-4}
         \cmidrule(lr){5-5} 
            & $\widehat{\alpha}$       
            & 
                \num{1.003 503}
            &
                (\num{0.920 568}, \num{1.090 649})
            &
                \num{1.9e-3}
        \\
            & 
                $\widehat{\sigma}$
            & 
                \num{0.050 373}
            &
                (\num{0.049 590}, \num{0.051 130})
            &
                \num{1.6e-7}        
        \\
            & 
                $\widehat{m}$
            & 
                \num{2.048 797}
            &
                (\num{1.518 122}, \num{2.748 484})
            &
                \num{1.0e-1}
        \\
        \\
        & \multicolumn{4}{l}{\textbf{Case 4:} $\alpha=2.5$, $\sigma=0.01$, $m=15$}
        \\ 
            \cmidrule(lr){2-5} 
            &
            & 
                \multicolumn{1}{c}{
                    \textbf{PE}
                }
            & 
                \multicolumn{1}{c}{
                    \textbf{95Q}
                }
            & 
                \multicolumn{1}{c}{
                    \textbf{MSE} 
                }
        \\
         \cmidrule(lr){3-3} 
         \cmidrule(lr){4-4}
         \cmidrule(lr){5-5} 
            & $\widehat{\alpha}$       
            & 
                \num{2.501 108 }
            &
                (\num{2.482 557}, \num{2.519 105})
            &
                \num{8.5e-5}
        \\
            & 
                $\widehat{\sigma}$
            & 
                \num{0.014 972}
            &
                (\num{0.014 865}, \num{0.015 083})
            &
                \num{3.1e-9}        
        \\
            & 
                $\widehat{m}$
            & 
                \num{15.024 19}
            &
                (\num{14.031 43}, \num{16.221 40})
            &
                \num{3.0e-1}
        \\
        \bottomrule
    \end{tabular}
\caption{ Average estimators, 95$\%$ quantiles and Mean squared error of the estimated parameter  obtained from 1
000 simulated continuous paths.}\label{tab:cont}
\end{table}


\subsection{Experiment 2: Discrete (incomplete) data}
In this section, we will conduct a numerical study on the case of 
{\it incomplete observation} of the solution. This refers to situations where
the number of recorded data is small, and the estimators discussed in Section
\ref{sec:con_cas} do not yield accurate results. Then,
as previously mentioned, we are utilizing the version of the EM algorithm proposed 
in Section \ref{sec:dis_case}.

First, we simulated $500$ paths in the time interval $[0,10]$. For each simulated path, 
we set the initial value $ X (0) = \num{0.05}$ and the parameter values 
$\alpha_0=\num{0.4}$, $ m_0=2$, and $ \sigma_0=\num{0.01}$. Each path has
been considered a data set with a $n=\num{10000}$ length. 
Out of each data set, we extracted incomplete data by retaining only $10\%$ of them,
which is equivalent to $k=\num{1000}$.

We run Algorithm \ref{al:em} for each incomplete sample using 10 iterations. In the E-step, we reconstruct the complete dataset using diffusion bridges based on a set of proposed parameters. This process allows us to generate a dataset with the same number of points as the original dataset. To achieve this, we generate 100 bridges between the points of the incomplete dataset. The initial proposed parameters are obtained by straightforwardly applying the method from Section \ref{sec:con_cas} to the incomplete dataset. Although these parameters may not be accurate, they can provide a first approximation of the actual values.

The average of estimators for each parameter was plotted in Figure \ref{fig:em} for every iteration of Algorithm \ref{al:em}.  Figure \ref{fig:em} illustrates the algorithm's simultaneous convergence for the three parameters. It is important to recall that the convergence of each parameter depends on the others. This is the reason why the initial iteration, marked by incomplete data, produces nonoptimal results. However, with subsequent iterations, the results gradually improve until they all converge to the true value at approximately the same iteration.
\begin{figure}[H]
  \centering
\includegraphics[width=0.9\textwidth]{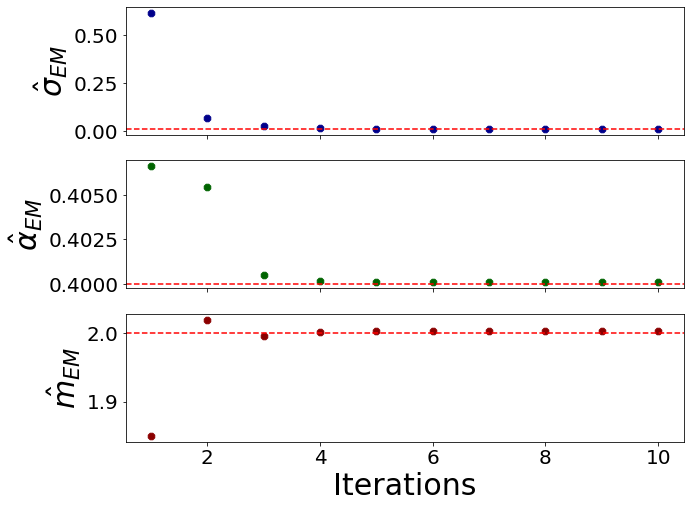}
    \caption{Average estimations for each iteration with $\alpha=0.4$, $m=2$, and $\sigma=0.01$.}
    \label{fig:em}

\end{figure}

On the other hand, we implemented Algorithm \ref{al:em} with three different sets of parameters and two different numbers of observations for each scenario. 
In Table \ref{tab:EM}, we present the averages and $95\%$ quantiles for a sample obtained from the $10$th iteration of Algorithm \ref{al:em} across 1,000 datasets. In particular, Table \ref{tab:EM} showcases the impressive resilience of the algorithm, revealing that the estimators derived from using merely 1\% of the complete data are remarkably akin to the estimators obtained via a sparse 10\% sample. Furthermore, it is evident that there are no biases in either scenario, and the results of experiments for the incomplete case with limited data (1\%) are highly satisfactory.
\begin{table}[H]
    \flushleft
    \textbf{Notation:}
    \\
        \hspace{1cm} \textbf{PIF} partial information fraction 
    \\
        \hspace{1cm} \textbf{EMPE} EM-puntual estimator
    \\
        \hspace{1cm} \textbf{MSE} mean squarred error
    
    \centering
    \begin{tabular}{
        @{}%
            w{l}{0.05\textwidth}%
            w{r}{0.15\textwidth}%
            w{r}{0.15\textwidth}%
            w{r}{0.15\textwidth}%
            w{r}{0.15\textwidth}%
            w{r}{0.15\textwidth}%
        @{}%
    }
        \toprule
        \\
            & \multicolumn{5}{l}{\textbf{Case 1:} $\alpha=0.7$, $\sigma=0.01$, $m=0.6$}
        \\ 
            \cmidrule(lr){2-6} 
            &                          
            & 
                \multicolumn{2}{c}{\textbf{PIF: 10\%}} 
            & 
                \multicolumn{2}{c}{\textbf{PIF: 1\%}} 
        \\ 
            \cmidrule(lr){3-4}
            \cmidrule(lr){5-6}
            &
            & 
                \multicolumn{1}{r}{
                    \textbf{EMPE}
                }
            & 
                \multicolumn{1}{r}{
                    \textbf{MSE}
                }
            & 
                \multicolumn{1}{r}{
                    \textbf{EMPE}
                }
            & 
                \multicolumn{1}{r}{
                    \textbf{MSE} 
                }
        \\
            & $\widehat{\alpha}$       
            & 
                \num{0.700923}
            &
                \num{6.1e-4}
            &
                \num{0.701065}
            &
                \num{6.1e-4}
        \\
            & 
                $\widehat{\sigma}$
            &
                \num{0.010032}
            &   
                \num{7.9e-8}
            &   
                \num{0.009341}
            &   \num{6.9e-7}            
        \\
            & 
                $\widehat{m}$
            &
                \num{0.600778}
            &
                \num{1.3e-3}
            &
                \num{0.600505}
            &
                \num{1.3e-3}
        \\
        \\
        &
            \multicolumn{5}{l}{ 
                \textbf{Case 2:}
                $\alpha=\num{0.9}$,
                $\sigma=\num{0.01}$,
                $m=\num{1}$
            }
        \\ 
            \cmidrule(lr){2-6} 
            &                          
            & 
                \multicolumn{2}{c}{\textbf{PIF: 10\%}} 
            & 
                \multicolumn{2}{c}{\textbf{PIF: 1\%}} 
        \\ 
            \cmidrule(lr){3-4}
            \cmidrule(lr){5-6}
            &
            & 
                \multicolumn{1}{r}{
                    \textbf{EMPE}
                }
            & 
                \multicolumn{1}{r}{
                    \textbf{MSE}
                }
            & 
                \multicolumn{1}{r}{
                    \textbf{EMPE}
                }
            & 
                \multicolumn{1}{r}{
                    \textbf{MSE} 
                }
     \\
        &
            $\widehat{\alpha}$
        &
            \num{0.902120}
        &
            \num{2.1e-4}
        &
            \num{0.902301}
        &
            \num{2.2e-4}
    \\
        &
            $\widehat{\sigma}$
        &
            \num{0.010424}
        &
            \num{5.9e-8}
        &
            \num{0.009679}
        &
            \num{5.1e-7}
    \\
        &
            $\widehat{m}$
        &
            \num{0.996465}
        &
            \num{1.3e-3}
        &
            \num{0.995809}
        &
            \num{1.3e-3}
        \\
        \\
        &
            \multicolumn{5}{l}{ 
                \textbf{Case 3:}
                $\alpha=\num{1}$,
                $\sigma=\num{0.05}$,
                $m=\num{2}$
            }
        \\ 
            \cmidrule(lr){2-6} 
            &                          
            & 
                \multicolumn{2}{c}{\textbf{PIF: 10\%}} 
            & 
                \multicolumn{2}{c}{\textbf{PIF: 1\%}} 
        \\ 
            \cmidrule(lr){3-4}
            \cmidrule(lr){5-6}
            &
            & 
                \multicolumn{1}{r}{
                    \textbf{EMPE}
                }
            & 
                \multicolumn{1}{r}{
                    \textbf{MSE}
                }
            & 
                \multicolumn{1}{r}{
                    \textbf{EMPE}
                }
            & 
                \multicolumn{1}{r}{
                    \textbf{MSE} 
                }
     \\
        &
            $\widehat{\alpha}$
        &
            \num{1.006274}
        &
            \num{1.9e-3}
        &
            \num{1.008769}
        &
            \num{1.9e-3}
    \\
        &
            $\widehat{\sigma}$
        &
            \num{0.047267}
        &
            \num{1.4e-6}
        &
            \num{0.041388}
        &
            \num{1.3e-5}
    \\
        &
            $\widehat{m}$
        &
            \num{2.015839}
        &
            \num{9.1e-2}
        &
            \num{1.976043}
        &
            \num{8.5e-2}
        \\
        \\
        &
            \multicolumn{5}{l}{ 
                \textbf{Case 4:}
                $\alpha=\num{2.5}$,
                $\sigma=\num{0.01}$,
                $m=\num{15}$
            }
        \\ 
            \cmidrule(lr){2-6} 
            &                          
            & 
                \multicolumn{2}{c}{\textbf{PIF: 10\%}} 
            & 
                \multicolumn{2}{c}{\textbf{PIF: 1\%}} 
        \\ 
            \cmidrule(lr){3-4}
            \cmidrule(lr){5-6}
            &
            & 
                \multicolumn{1}{r}{
                    \textbf{EMPE}
                }
            & 
                \multicolumn{1}{r}{
                    \textbf{MSE}
                }
            & 
                \multicolumn{1}{r}{
                    \textbf{EMPE}
                }
            & 
                \multicolumn{1}{r}{
                    \textbf{MSE} 
                }
     \\
        &
            $\widehat{\alpha}$
        &
            \num{2.503135}
        &
            \num{8.6e-5}
        &
            \num{2.502366}
        &
            \num{8.8e-5}
    \\
        &
            $\widehat{\sigma}$
        &
            \num{0.013377}
        &
            \num{1.3e-5}
        &
            \num{0.011455}
        &
            \num{6.5e-5}
    \\
        &
            $\widehat{m}$
        &
            \num{14.08557}
        &
            \num{2.2e-4}
        &
            \num{13.69698}
        &
            \num{3.7e-4}
        \\
        \bottomrule
    \end{tabular}
    \caption{EMPE estimation}
    \label{tab:EM}
\end{table}

We are interested in perform a comparison between the accuracy of the estimators for the scenarios considered in this work, namely with complete observation at $[0,10]$ and incomplete data with 10\% and 1\% of the complete information. Thus, we performed simulations with two different combinations of parameters in particular, one for $m<1$ and other for $m>1$. In Table \ref{tab:comp}, we report the estimators of the parameters, which is obtained using the average of 1000 databases and their corresponding mean squared error for the three data scenarios. The results for the  incomplete cases were obtained by running Algorithm \ref{al:em} up to 10 iterations. The numerical results demonstrate the consistency of the estimators based on the available information in each scenario. In other words, there is improved precision with greater information. Nevertheless, it is important to highlight that even in the extreme case of the third scenario, the estimation method presented in this paper performs exceptionally well.
\begin{table}[H]
    \flushleft
    \textbf{Notation:}
    \\
       \hspace{1cm} \textbf{CI} complete information
    \\
        \hspace{1cm} \textbf{PIF} partial information fraction 
        \\
        \hspace{1cm} \textbf{PE} puntual estimator
    \\
        \hspace{1cm} \textbf{EMPE} EM-puntual estimator
    \\
        \hspace{1cm} \textbf{MSE} mean squarred error
    
    \centering
    \begin{tabular}{
        @{}%
            w{r}{0.1\textwidth}%
            w{r}{0.13\textwidth}%
            w{r}{0.13\textwidth}%
            w{r}{0.13\textwidth}%
            w{r}{0.13\textwidth}%
            w{r}{0.13\textwidth}%
            w{r}{0.13\textwidth}%
        @{}%
    }
        \toprule
        \\
            & \multicolumn{5}{l}{\textbf{Case 1:} $\alpha=0.7$, $\sigma=0.01$, $m=0.6$}
        \\ 
            \cmidrule(lr){1-7} 
            &  \multicolumn{2}{c}{\textbf{CI}}&  \multicolumn{2}{c}{\textbf{PIF: 10\%}} &  \multicolumn{2}{c}{\textbf{PIF: 1\%}} 
        \\ 
            \cmidrule(lr){2-3}
            \cmidrule(lr){4-5}
            \cmidrule(lr){6-7}
            & 
                \multicolumn{1}{c}{
                    \quad\textbf{PE}
                }
            & 
                \multicolumn{1}{c}{
                    \textbf{MSE}
                }
            & 
                \multicolumn{1}{r}{
                    \textbf{EMPE}
                }
            & 
                \multicolumn{1}{c}{
                    \textbf{MSE} 
                }
                & 
                \multicolumn{1}{r}{
                    \textbf{EMPE}
                }
            & 
                \multicolumn{1}{c}{
                    \textbf{MSE} 
                }
        \\
            $\widehat{\alpha}$     
            & 
                \num{0.700 888}
            &
                \num{6.1e-4}
            & 
                \num{0.700 923}
            &
                \num{6.1e-4}
            &
                \num{0.701 065}
            &
                \num{6.1e-4}
        \\
                $\widehat{\sigma}$
            & 
                \num{0.011 282}
            &
                \num{8.5e-9}
            & 
                \num{0.010 032 }
            &
                \num{7.9e-8}
            &
                \num{0.009 341}
            &
                \num{6.9e-7}        
        \\
                $\widehat{m}$
            & 
                \num{0.600 990}
            &
                \num{1.3e-3}
            & 
                \num{0.600 778}
            &
                \num{1.3e-3}
            &
                \num{0.600 505}
            &
                \num{1.3e-3}
        \\
        \\
        & \multicolumn{5}{l}{\textbf{Case 3:} $\alpha=1$, $\sigma=0.05$, $m=2$}
        \\ 
            \cmidrule(lr){1-7} 
            &  \multicolumn{2}{c}{\textbf{CI}}&  \multicolumn{2}{c}{\textbf{PIF: 10\%}} &  \multicolumn{2}{c}{\textbf{PIF: 1\%}} 
        \\ 
            \cmidrule(lr){2-3}
            \cmidrule(lr){4-5}
            \cmidrule(lr){6-7}
            & 
                \multicolumn{1}{c}{
                    \quad\textbf{PE}
                }
            & 
                \multicolumn{1}{c}{
                    \textbf{MSE}
                }
            & 
                \multicolumn{1}{r}{
                    \textbf{EMPE}
                }
            & 
                \multicolumn{1}{c}{
                    \textbf{MSE} 
                }
                & 
                \multicolumn{1}{r}{
                    \textbf{EMPE}
                }
            & 
                \multicolumn{1}{c}{
                    \textbf{MSE} 
                }
        \\
            $\widehat{\alpha}$       
            & 
                \num{1.003 503}
            &
                \num{1.9e-3}
            & 
                \num{1.006 274}
            &
                \num{1.9e-3}
            &
                \num{1.008 769}
            &
                \num{1.9e-3}
        \\
                $\widehat{\sigma}$
            & 
                \num{0.050 373}
            &
                \num{1.6e-7}
            & 
                \num{0.047 267}
            &
                \num{1.4e-6}
            &
                \num{0.041 388}
            &
                \num{1.3e-5}           
        \\
                $\widehat{m}$
           & 
                \num{2.048 797}
            &
                \num{1.0e-1}
            & 
                \num{2.015 839}
            &
                \num{9.1e-2}
            &
                \num{1.976 043}
            &
                \num{8.5e-2}        
        \\
        \bottomrule
    \end{tabular}
    \caption{Estimation comparison}
    \label{tab:comp}
\end{table}

\section{Discussion-Conclusion} \label{sec-Conclusions}

\paragraph{Statement of principal findings}
The first aim of this paper was to develop consistent estimators for the
parameters $\sigma$, $\alpha$, and $m$. Our second goal was to address the issue of 
incomplete information. Building upon these objectives, we present the following arguments.

\paragraph{Strengths and weaknesses of the study}
The estimators for $\alpha$ and $m$ presented in this study are proven to be consistent according to 
Theorem \ref{theo:cons}. This means that as the simulation time increases, more information becomes 
available, and the accuracy of our estimates improves significantly.
Nevertheless, it is important to note that satisfactory results can be achieved even with a modest amount 
of information. Indeed, Figure \ref{fig:em} illustrates the convergence of our estimates to the true value 
after only 10 iterations, using a small initial data set (1\%) obtained from Experiment 1 (see subsection 
\ref{sub_sec:experimtn01}). The success achieved with incomplete information can be attributed to using the 
EM algorithm, which enables us to obtain accurate estimates for the three parameters involved in the model 
from this work.
 During the numerical experiments, we found that if we assume that $m\in (0,\tfrac{1}{2}]$,
 then the method we use to find the MLE for $m$ fails. This could be due to the fact that the 
 involved function is $\gamma$-H\"older continuous with $\gamma\le m<\tfrac{1}{2}$, and therefore, 
 the numerical method we used becomes unstable. Thus, it could be necessary to implement
 a different method to estimate parameter $m$ for these cases.

\paragraph{Related Work}
    To provide some context related to our findings, we mention the contributions of
    J.C. Cortés et al. \cite{Cortes2019} and V. Bevia et al. \cite{Bevia2023}. 
    In \cite{Bevia2023}, the authors conducted a thorough analysis by utilizing 
    the same deterministic base but as a random differential equation. In \cite{Cortes2019}, 
    the authors provide a comprehensive analysis via Karhunen-Loeve expansion. To the best of our
    knowledge, these two contributions are among the few that examine and calibrate the 
    randomized version of Richardson's Logistic Model. 

\paragraph{Meaning of the study}    
        In the literature mentioned, we have made a contribution by analyzing, calibrating,
    and treating partial information of a stochastic version of the generalized logistic DE.
    As far as we know, our work is the first  attempt to estimate parameters from a  
    stochastic version of a generalized Logistic model as in \Cref{SGLDE}, using a version of 
    EM algorithm while also addressing partial information. 
    Our approach involves a rigorous analysis of the stochastic version of the generalized 
    logistic DE and its partial information. We employed the EM algorithm and showed
    its performance in estimating the model parameters
    with partial information. In particular, Algorithm \ref{al:em} works well, even if the distance between each data is different.

\paragraph{Unanswered questions and future research}   
   Our methodology relies on formal stochastic analysis to deduce consistent estimators, so our simulations 
   suggest robust and efficient implementations to work with real data. However, we believe that the amount 
   of information that implies working with real data and other issues leads to our next publication. The 
   asymptotic normality of the MLE remains to be studied. This work opens the door to applying the SGLDE to 
   actual data. Indeed, we will strive to utilize the models and methodologies outlined in this study to 
   analyze real-world data, for instance, capture data in the fishery  (cf. \cite{del-bal}).

\section*{Acknowledgements} The research was supported by PAPIIT-IN102224.

\appendix

\section{Proof of consistency}\label{Appe_1}

This section is devoted to the proof of Theorem \ref{theo:cons}.

First let us to review some facts that will be useful for the proof. We denote $B_t=B(t)$ and $X_t=X(t)$, and
we define the following expressions:

\begin{align*}
   I_1&:= I_1(T):= \int_0^T X_t^{\hat{m}} \, \ln(X_t)\,  dB_t,\\
    I_2&:= I_2(T):= \int_0^T (1-X_t^{\hat{m}})\,  dB_t,\\
    J_1&:= J_1(T):= \int_0^T X_t^{\hat{m}} \, \ln(X_t)\,  dt,\\
    J_2&:= J_2(T):= \int_0^T (1-X_t^{\hat{m}})^2\,   dt,\\
    J_3&:= J_3(T):= \int_0^T X_t^{2\hat{m}} \, \ln(X_t)\,  dt,\\ 
   \mathcal{J}_1&:= \mathcal{J}_1(T):= \int_0^T X_t^{\hat{m}} (X_t^{m_0}-X_t^{\hat{m}} ) \, \ln(X_t)\,  dt,\\
\mathcal{J}_2&:=\mathcal{J}_2(T):= \int_0^T (1-X_t^{\hat{m}}) (X_t^{\hat{m}}-X_t^{m_0}) \,  dt.
\end{align*}
Note that, by the properties of the stochastic solution $X_t$, we have that $J_1<0$, $J_2>0$ and $J_3<0$, with probability $1$.

\begin{proof}[Proof of Theorem \ref{theo:cons}]

We use the SDE \eqref{SGLDE} into the equation \eqref{MLE-alpha} and \eqref{eq:m}, and after some algebra we obtain two equations as follows

\begin{align*}
    J_2\big( \hat{\alpha} - \alpha_0\big) - 
 \mathcal{J}_2\alpha_0 & = \sigma I_2,\\
 (J_1+ J_3) \big( \hat{\alpha} - \alpha_0\big) - \mathcal{J}_1\alpha_0 & = \sigma I_1.
\end{align*}
    
From these equations we can get
\begin{align*}
    J_2\big( \hat{\alpha} - \alpha_0\big) - \alpha_0 
 (\mathcal{J}_1+ \mathcal{J}_2)  & = \sigma I_2 - \alpha_0 \mathcal{J}_1 ,\\
 (J_1+ J_3) \big( \hat{\alpha} - \alpha_0\big) - \alpha_0 
 (\mathcal{J}_1+ \mathcal{J}_2)& = \sigma I_1 - \alpha_0 \mathcal{J}_2.
\end{align*}

Then, we write the last two equations as a matrix system

\begin{align}
    \begin{pmatrix}
         (J_1+ J_3) & -\alpha_0 \\
        J_2 & -\alpha_0
    \end{pmatrix} \begin{pmatrix}
        \hat{\alpha} - \alpha_0\\
        \mathcal{J}_1+ \mathcal{J}_2
    \end{pmatrix} = 
    \begin{pmatrix}
        \sigma I_1 - \alpha_0 \mathcal{J}_2\\
        \sigma I_2 - \alpha_0 \mathcal{J}_1 
    \end{pmatrix},
\end{align}

Which can be solved explicitly:

\begin{align}\label{eq:alpha+hidden_m}
    \begin{pmatrix}
        \hat{\alpha} - \alpha_0\\
        \mathcal{J}_1+ \mathcal{J}_2
    \end{pmatrix} &= \frac{1}{\alpha_0\big(-J_1+J_2- J_3\big)} \begin{pmatrix}
         -\alpha_0 & \alpha_0 \\
        -J_2 & (J_1+ J_3)
    \end{pmatrix} 
    \begin{pmatrix}
        \sigma I_1 - \alpha_0 \mathcal{J}_2\\
        \sigma I_2 - \alpha_0 \mathcal{J}_1 
    \end{pmatrix}\nonumber
    \\
    &= \frac{1}{\alpha_0\big(-J_1+J_2- J_3\big)} \begin{pmatrix}
        \alpha_0 \big(-\sigma I_1 + \alpha_0 \mathcal{J}_2 +
        \sigma I_2 - \alpha_0 \mathcal{J}_1 \big) \\
        \,
        \\
       -J_2\big[ \sigma I_1 - \alpha_0 \mathcal{J}_2 \big] + (J_1+ J_3)\big[ 
        \sigma I_2 - \alpha_0 
        \mathcal{J}_1\big] 
    \end{pmatrix}    
\end{align}
We first will study the equation for $\mathcal{J}_1+ \mathcal{J}_2$:
\begin{align*}
   \mathcal{J}_1+ \mathcal{J}_2 &= \frac{1}{\alpha_0\big(-J_1+J_2- J_3\big)}  \Big( -J_2\big[ \sigma I_1 - \alpha_0 \mathcal{J}_2 \big] + (J_1+ J_3)\big[ 
        \sigma I_2 - \alpha_0         \mathcal{J}_1\big]\Big)\\
        &= \sigma \frac{(J_1+J_3) I_2 -J_2 I_1 }{\alpha_0\big(-J_1+J_2- J_3\big)} + \frac{J_2\mathcal{J}_2 - (J_1+J_3) \mathcal{J}_1 }{\big(-J_1+J_2- J_3\big)},
\end{align*}
and by regrouping the terms we have
\begin{align}
  \mathcal{J}_1 \, \Big[ 1 + \frac{J_1+J_3}{\big(-J_1+J_2- J_3\big)}\Big] &+ \mathcal{J}_2   \, \Big[ 1 - \frac{J_2}{\big(-J_1+J_2- J_3\big)}\Big] \nonumber\\
  &= \sigma \frac{(J_1+J_3)  }{\alpha_0\big(-J_1+J_2- J_3\big)} \, I_2 - \sigma \frac{J_2  }{\alpha_0\big(-J_1+J_2- J_3\big)}\, I_1\nonumber\\
  &=: \frac{\sigma}{\alpha_0} Z_2 \, I_2 - \frac{\sigma}{\alpha_0} Z_1 \, I_1 .
\end{align}
It is possible to show that $\E(Z_1)$ and $\E(Z_2)$ are both finite and in addition, $\E(I_1)=\E(I_2)=0$ for all $T>0$, this implies in particular that, for  $k=1,2$, $I_k$ converges to zero in probability, when $T\rightarrow \infty$. We now will show that, for  $k=1,2$, $Z_k I_k$ converges, in $L^1$, to zero, and thus in probability when $T \rightarrow \infty$.
First, observe that with probability one, the random variable $|Z_k|$ is bounded by 1, thus we have that $Z_k I_k < I_k$, thus, we have that 
$$
\Pb\big(Z_k I_k > b \big) = \Pb\big(I_k > b \big) \le \frac{1}{b} \E(I_k) =0
$$
for all $b>0$ and all $T>0$.
This implies, that 

$$
\Pb\big(Z_k I_k \le b \big) =1, \quad \mbox{for all $T>0$ and $b>0$}.
$$

Thus, by using Theorem 4 in page 354 of \cite{{Gr-St-20}}, we conclude that $Z_k I_k$ converges to zero, in $L^1(\Omega)$, when $T\rightarrow \infty$. In particular, it converges to zero in probability.
Therefore, we have showed that for all $T>0$
\begin{align*}
\E\Bigg(  \mathcal{J}_1 \, \Big[ 1 + \frac{J_1+J_3}{\big(-J_1+J_2- J_3\big)}\Big] &+ \mathcal{J}_2   \, \Big[ 1 - \frac{J_2}{\big(-J_1+J_2- J_3\big)}\Big] \Bigg)= 0
\end{align*}
But, since the last expectation is zero for all $T>0$, thus, it must be zero for both expressions, meaning we have for all $T>0$
\begin{align}\label{unbiased_m0}
\E\Bigg(  \mathcal{J}_1 \, \Big[ 1 + \frac{J_1+J_3}{\big(-J_1+J_2- J_3\big)}\Big]\Bigg)= 0, \quad  \E\Bigg( \mathcal{J}_2   \, \Big[ 1 - \frac{J_2}{\big(-J_1+J_2- J_3\big)}\Big] \Bigg)&= 0,
\end{align}
and by recalling that the denominator $-J_1+J_2- J_3 >0 $ we deduce that 
\begin{align}\label{unbiased_m}
    \E \big( \mathcal{J}_1 )=0, \quad \E \big( \mathcal{J}_2 )=0,\quad \mbox{for all $T>0$ }
\end{align}
Then, in particular, we have that 

$$
\E\Big(\mathcal{J}_2(T)\Big)= \E\Big( \int_0^T (1-X_t^{\hat{m}}) (X_t^{m_0}- X_t^{\hat{m}}) \,  dt\Big)=0, \mbox{for all $T>0$ }
$$

Thus, from this we deduce that $\hat{m}$ converges to $m_0$ in $L^1$ and thus in probability, when $T \rightarrow \infty$.

We now prove the consistency for the parameter $\alpha$. We have that 

\begin{align}\label{alpha_decom}
    \hat{\alpha} - \alpha_0 &= \frac{1}{\big(-J_1+J_2- J_3\big)}   \big(-\sigma I_1 + \alpha_0 \mathcal{J}_2 +
        \sigma I_2 - \alpha_0 \mathcal{J}_1 \big)\nonumber\\
        &= \frac{  \sigma }{\big(-J_1+J_2- J_3\big)}   \big(I_2 -   I_1 \big)  + \frac{\alpha_0 }{\big(-J_1+J_2- J_3\big)}  \big(  \mathcal{J}_2  - \mathcal{J}_1 \big)
\end{align}
We know that $\big(-J_1+J_2- J_3\big)>0$, thus the denominator is bounded for all $T>0$. This implies that there exists a positive constant $C>0$ such that  
$$
\sup_{T>0}\frac{1}{\big(-J_1+J_2- J_3\big)}  <C \quad\mbox{with probability $1$}.
$$

Thus, for any $b>0$ and $T>0$
\begin{align*}
\Pb\Bigg(    \frac{  \sigma }{\big(-J_1+J_2- J_3\big)}   \big(I_2 -   I_1 \big)\, > b \Bigg) \le \frac{1}{b}  \E\Bigg(    \frac{  \sigma }{\big(-J_1+J_2- J_3\big)}   \big(I_2 -   I_1 \big)\Bigg) = 0 
\end{align*}

Thus, by using again the Theorem 4 in page 354 of \cite{{Gr-St-20}}, we conclude that 
$$
\frac{  \sigma }{\big(-J_1+J_2- J_3\big)}   \big(I_2 -   I_1 \big) 
$$
 converges to zero, in $L^1(\Omega)$, when $T\rightarrow \infty$.

Now, by recalling the convergence \eqref{unbiased_m0} we conclude that the second term in right hand of \eqref{alpha_decom} converges to zero, in $L^1(\Omega)$, when $T\rightarrow \infty$. This implies the consistency of $\hat{\alpha}$.
 
\end{proof}

\end{document}